\documentclass[a4paper,UKenglish,pdfa,thm-restate,numberwithinsect]{article}

\usepackage[utf8]{inputenc}
\usepackage{fullpage}
\usepackage{enumitem}
\usepackage{amsmath, amsthm, amssymb, mathtools, xcolor, framed, tikz}
\usepackage{thmtools}
\usepackage{algorithm}
\usepackage{algpseudocode}
\usepackage{comment}
\usetikzlibrary{trees,arrows.meta,positioning,shapes.multipart}
\usetikzlibrary{positioning}
\usetikzlibrary{shapes}
\usetikzlibrary{decorations.pathreplacing}
\usetikzlibrary{calc,patterns,angles,quotes}
\usetikzlibrary{intersections}
\usepackage{hyperref}
\usepackage{float}
\hypersetup{colorlinks=true, citecolor=blue}
\usepackage[capitalise]{cleveref}
\usepackage{soul}
\usepackage{cite}
\usepackage{mathabx}
\usepackage{regexpatch}

\makeatletter
\@addtoreset{ALG@line}{algorithm}
\makeatother

\newcommand{\Dcc}{\textrm{D}^{cc}}
\newcommand{\Rcc}{\textrm{R}^{cc}}
\newcommand{\Qcc}{\textrm{Q}^{cc}}

\newcommand{\bs}{\textrm{bs}}

\newcommand{\mbs}{\textrm{MBS}}

\newcommand{\spar}{\textrm{spar}}

\newcommand{\adeg}{\widetilde{\textrm{deg}}}
\newcommand{\arank}{\widetilde{\textrm{rank}}}
\newcommand{\rank}{\textrm{rank}}

\renewcommand{\deg}{\textrm{deg}}

\newcommand{\agamma}{\widetilde{\gamma_2}}
\newcommand{\sampleliftedrho}
{\textsc{LiftedRestriction}}

\newcommand{\OR}{\mbox{{\sc Or}}}
\newcommand{\pOR}{\mbox{{\sc promise-Or}}}
\newcommand{\AND}{\mbox{{\sc And}}}
\newcommand{\ANDOR}{\mbox{{\sc And-Or}}}
\newcommand{\XOR}{\mbox{{\sc Xor}}}
\newcommand{\EQ}{\mbox{{\sc EQ}}}
\newcommand{\UDISJ}{\mbox{{\sc UDISJ}}}

\renewcommand{\epsilon}{\varepsilon}

\newcommand{\lift}{\textrm{Lift}}

\newcommand{\fourierwt}[1]{\left\|\widehat{#1}\right\|_1}
\newcommand{\fsparsity}[1]{\left\|\widehat{#1}\right\|_0}

\declaretheorem[name=Theorem,numberwithin=section]{theorem}

\declaretheorem[name=Lemma,sibling=theorem]{lemma}
\declaretheorem[name=Claim,sibling=theorem]{claim}

\declaretheorem[name=Definition,sibling=theorem]{definition}
\declaretheorem[name=Example,sibling=theorem]{example}
\declaretheorem[name=Remark,sibling=theorem]{remark}

\declaretheorem[name=Conjecture,sibling=theorem]{conjecture}

\makeatletter
\xpatchcmd\thmt@restatable{%
  \csname #2\@xa\endcsname\ifx\@nx#1\@nx\else[{#1}]\fi
}{%
  \ifthmt@thisistheone
    \csname #2\@xa\endcsname\ifx\@nx#1\@nx\else[{#1}]\fi
  \else
    \csname #2\@xa\endcsname[{Restated}]
  \fi
}{}{}
\makeatother

\title{Quantum–Classical Equivalence for \AND-Functions}

\author{
Sreejata Kishor Bhattacharya
\thanks{TIFR, Mumbai. Email: {\tt sreejata.bhattacharya@tifr.res.in}. Supported by the Department of Atomic Energy, Govmt. of India, under project \#RTI4014 and by a Google PhD Fellowship.}
\and
Farzan Byramji
\thanks{UC San Diego. Email: {\tt fbyramji@ucsd.edu}. Supported by a Simons Investigator Award \#929894, and NSF awards CCF-2425349 and AF: Medium 2212136.}
\and
Arkadev Chattopadhyay
\thanks{TIFR, Mumbai. Email: {\tt arkadev.c@tifr.res.in}. Supported by the Department of Atomic Energy, Govmt. of India, under project \#RTI4014 and by a Google India Faculty Award.} 
\and 
Yogesh Dahiya 
\thanks{UC San Diego. Email: {\tt ydahiya@ucsd.edu}. Supported by Simons Investigator Award \#929894 and NSF award CCF-2425349.} 
\and 
Shachar Lovett 
\thanks{UC San Diego. Email: {\tt slovett@ucsd.edu}. Supported by Simons Investigator Award \#929894 and NSF award CCF-2425349.} 
}
\date{}

\begin{document}
\maketitle

\begin{abstract}
A major open problem in quantum communication complexity is whether
quantum protocols can be exponentially more efficient than classical protocols for
computing total Boolean functions; the prevailing conjecture is that they cannot be so.
In a seminal work, Razborov (2002) resolved this question for $\AND$-functions of the form
\[
F(x,y) = f(x_1 \land y_1, \ldots, x_n \land y_n),
\]
when the outer function \( f \) is symmetric, by proving that their bounded-error
quantum and classical communication complexities are polynomially related.
Since then, extending this result to \emph{all} $\AND$-functions has remained open
and has been posed by several authors.

In this work, we settle this problem in a strong way. We show that for every Boolean function \( f \),
the bounded-error quantum and classical deterministic communication complexities of the
function \( f \circ \AND_2 \) are polynomially related, up to polylogarithmic
factors in \( n \). We prove this by showing that
both are characterized—up
to polynomial loss—by the logarithm of the De Morgan sparsity of \( f \).

Our results build on the recent work of Chattopadhyay, Dahiya, and Lovett \cite{chattopadhyay2026restriction} on structural
characterizations of non-sparse Boolean functions, which we extend to
resolve the conjecture for general $\AND$-functions.
\end{abstract}

\section{Introduction}

Understanding when/if the laws of quantum mechanics can be exploited to provide significant computational advantage is a major research theme. Two classical results show that indeed quantum algorithms can have surprising power. Shor's polynomial time algorithm on quantum computers for integer factoring beats all \emph{known} classical factoring algorithms by an \emph{exponential} speed-up. Grover search in the quantum query model beats all \emph{possible} classical algorithms by a \emph{quadratic} factor. The first result on exponential speed-up is \emph{conditional} on the hardness of factoring by classical computers. The second result is provable advantage but only quadratic and we know no better is possible in that query model. In this work, we're concerned with provable exponential advantage for quantum models when the task is to compute a total Boolean function. It is well known that both the factoring problem and the search problem for a marked item that Grover's algorithm solves, can readily be re-phrased in terms of total Boolean functions. 

Communication complexity, introduced by Yao~\cite{yao1979some}, studies the amount of
communication required to compute a function whose input is distributed among
multiple parties. Since its inception, communication complexity has become a central tool in
theoretical computer science, with applications ranging from streaming algorithms
and time--space tradeoffs to data structure lower bounds and circuit complexity. See the textbooks~\cite{KN-CC-book,rao2020communication} for excellent introductions to the area and its applications.

In the most standard setting—the two-party model—two players, Alice and Bob,
wish to compute a Boolean function
\( F : X \times Y \to \{0,1\} \),
where Alice receives \( x \in X \) and Bob receives \( y \in Y \).
They exchange messages according to a pre-agreed protocol in order to compute
\( F(x,y) \), with the goal of minimizing the total number of bits communicated in
the worst case.

Several variants of communication complexity arise depending on the type of
interaction allowed and whether the protocol is permitted to err with small
probability. In this work, we focus on three such models: two classical models and one quantum
model.

In the classical deterministic model, the \emph{deterministic communication
complexity} of \( F \), denoted \( \Dcc(F) \), is the minimum number of bits that
must be exchanged by a protocol that computes \( F(x,y) \) correctly on all inputs.
In the public-coin randomized model, Alice and Bob have access to shared public
randomness and are required to compute \( F(x,y) \) with error probability at most
\( 1/3 \) on every input; the corresponding complexity measure is the
\emph{randomized communication complexity}, denoted \( \Rcc(F) \).

The quantum variant of communication complexity, also introduced by
Yao~\cite{yao1993quantum}, allows the messages exchanged between the parties to be
quantum states (qubits). At the end of the protocol, one of the parties performs a measurement on its
quantum state to produce the output. In addition, the parties may share an arbitrary entangled state prior to the start
of the protocol, at no communication cost~\cite{cleve1997substituting}. The number of qubits exchanged by the best such protocol that computes
\( F(x,y) \) with error probability at most \( 1/3 \), for the worst input, is called the \emph{bounded-error quantum communication complexity} of \( F \), and is denoted
\( \Qcc(F) \).

A central question in quantum computation is to understand when quantum models
enable efficient solutions to problems that are believed to be hard for classical
models. In the context of communication complexity, this question asks when
quantum protocols can be super-polynomially more efficient than classical
(randomized) protocols, and has been a major driving force behind research in
the area.

Exponential quantum advantages were discovered early on for structured problems arising from
partial functions, sampling problems and relations, 
see \cite{raz1999exponential,buhrman2001quantum,ambainis2003quantum,gavinsky2006bounded,RegevK11}. This understanding has been refined over a number of papers, see \cite{Gavinsky20,G20,GirishRT21}, with latest work exhibiting exponential quantum advantage over such promised problems even when the quantum protocol is extremely restricted and no restrictions are imposed on competing classical protocols. In a very recent work, exponential quantum advantage was shown for a total search problem in TFNP by Goos et. al. \cite{GoosG0L25}.
However, as re-iterated by \cite{GoosG0L25}, no such separation is known for \emph{total} Boolean
functions which remains the holy grail of the search for quantum advantage. The largest gap currently known between bounded-error quantum and randomized communication
complexity for a total function is only polynomial: a cubic separation (up to
polylogarithmic factors), obtained by lifting cubic separations between quantum
and randomized query complexity \cite{bansal2021k,sherstov2021optimal}.

These results have led to the prevailing belief that, in the absence of
special structure—most notably for total Boolean functions—quantum and randomized
communication complexities are always polynomially related. This belief was
explicitly formulated by Shi and Zhu~\cite{SZ09} as the
\emph{Log-Equivalence Conjecture} (LEC).

\begin{conjecture}[Log-Equivalence Conjecture (LEC) \cite{SZ09}]
For every total Boolean function, the bounded-error quantum and randomized
communication complexities are polynomially related in the two-party
communication model.
\end{conjecture}

Despite intensive research efforts, the problem of proving or disproving the LEC remains wide open.
In light of the lack of progress on general total functions, several authors \cite{BuhrmanW01,razborov2003quantum,Klauck07,SZ09} have
proposed studying \emph{composed} communication problems of the form
\( F = f \circ g \), where \( g : \{0,1\}^b \times \{0,1\}^b \to \{0,1\} \) is a small,
preferably constant-size, gadget.
Among such restricted classes, $\AND$-functions of the form
\[
F(x,y) = f(x_1 \land y_1, \ldots, x_n \land y_n)
\]
have received particular attention.
This is partly motivated by the fact that some of the most studied problems in
communication complexity, such as \emph{Set Disjointness} and \emph{Inner Product},
are $\AND$-functions.

To the best of our knowledge, $\AND$-functions, especially in the context of quantum communication complexity, were first systematically studied by Buhrman and de Wolf~\cite{BuhrmanW01}. They showed that for this class, deterministic communication complexity and
zero-error quantum communication complexity are polynomially related when the
outer function \( f \) is symmetric or monotone.
However, the relationship between bounded-error quantum and randomized
communication complexity for $\AND$-functions remained poorly understood at the time.
Indeed, very few lower bounds were known against bounded-error quantum protocols,
essentially limited to those—such as Inner Product—obtained via the discrepancy method~\cite{kremer1995quantum}. In particular, no polynomial lower bound on the quantum communication complexity
of Set Disjointness was known. 

In a major breakthrough, Razborov~\cite{razborov2003quantum} established the optimal
\( \Omega(\sqrt{n}) \) lower bound for the bounded-error quantum communication complexity of
Set Disjointness over a universe of size \( n \). More generally, his method yielded that for 
$\AND$-functions \( f \circ \AND_2 \) with symmetric outer function
\( f \), the bounded-error quantum communication complexity is polynomially
equivalent even to the deterministic communication complexity.
Extending this result to \emph{all} $\AND$-functions remained open since then
and this has been investigated by several authors \cite{BuhrmanW01,razborov2003quantum,Klauck07,sherstov2008pattern,SZ09,Sherstov10}. Our main result resolves this problem.

\begin{restatable}{theorem}
{lecandfns}\label{thm:lec-and-fns}
Let \( f : \{0,1\}^n \to \{0,1\} \) be any Boolean function. Then:
\begin{enumerate}
\item
\(
\Dcc(f \circ \AND_2)
= O\!\left(\Qcc(f \circ \AND_2)^{7} \cdot (\log n)^2\right).
\)

\item
\(
\Dcc(f \circ \AND_2)
= O\!\left(\Rcc(f \circ \AND_2)^{5} \cdot (\log n)^2\right).
\)
\end{enumerate}
\end{restatable}

Ignoring polylogarithmic factors in \( n \), this shows that for all
$\AND$-functions, deterministic communication complexity is polynomially
equivalent to bounded-error quantum communication complexity. Notably, prior to our work it was not even known whether randomized
communication complexity is polynomially equivalent to deterministic complexity
for this class of functions.
Before describing our results, we briefly review prior developments related
to this question.

While not stated explicitly in Razborov’s work, his quantum lower bound for Set
Disjointness can be used to obtain lower bounds for a broader class of
composed functions. In particular, consider any constant-size gadget
\( g : \{0,1\}^b \times \{0,1\}^b \to \{0,1\} \)
whose communication matrix contains both $\AND_2$ and $\OR_2$ as submatrices.
For such gadgets, the bounded-error quantum communication complexity of
\( f \circ g \) is
\( \Omega(\sqrt{\bs(f)}) \) \cite{zhang2009tightness}, where \( \bs(f) \) denotes the block sensitivity of
\( f \).
This follows from the fact that a promised Set Disjointness instance of size
\( \bs(f) \) can be embedded into \( f \circ g \).

Combining this lower bound with Nisan’s classical result~\cite{Nisan91}, which
upper bounds the deterministic query complexity of \( f \) by a polynomial in
\( \bs(f) \), yields a quantum–classical equivalence for composed functions
\( f \circ g \) whenever \( g \) embeds both $\AND_2$ and $\OR_2$.
We refer to this class as \emph{$\ANDOR$-functions}.

Subsequently, using different techniques, Sherstov~\cite{sherstov2008pattern}
gave an independent proof of quantum–classical equivalence for $\ANDOR$-functions
via his pattern matrix method.
In related and independent work, Shi and Zhu~\cite{SZ09} proved quantum–classical equivalence for
composed functions with gadgets satisfying certain pseudorandomness properties,
although their gadgets were required to have size \( \Omega(\log n) \).

Another important line of work concerns $\XOR$-functions, i.e., functions composed
with the $\XOR_2$ gadget. Shi and Zhang~\cite{zhang2009communication} showed that,
up to polylogarithmic factors, the Log-Equivalence Conjecture holds for
$\XOR$-functions when the outer function \( f \) is symmetric. Subsequently,
Montanaro and Osborne~\cite{montanaro2009communication} proved a polynomial
equivalence between deterministic and zero-error quantum communication complexity
for $\XOR$-functions with monotone outer functions. With the $\AND$-function case
resolved in this work, extending the Log-Equivalence Conjecture to $\XOR$-functions
in full generality emerges as a natural next milestone toward a complete
understanding of the conjecture.

More recently, Chattopadhyay, Dahiya, and Lovett~\cite{chattopadhyay2026restriction}
revisited the Log-Equivalence Conjecture. In all previously known classes of
functions satisfying the conjecture, an even stronger statement holds: deterministic
communication complexity is polynomially related to bounded-error quantum
communication complexity. In~\cite{chattopadhyay2026restriction}, the authors
studied composed functions of the form \( f \circ \EQ_4 \), where \( \EQ_4 \)
denotes equality on four bits. Since \( \AND \circ \EQ_4 \) (which is just the equality function) belongs to this class,
deterministic communication complexity is exponentially separated from randomized
communication complexity; nevertheless, they showed that the Log-Equivalence
Conjecture continues to hold (up to polylogarithmic factors). Our work builds on and
extends the ideas developed in~\cite{chattopadhyay2026restriction}, which we
elaborate on in subsequent sections.

\subsection{Our Results}
For any Boolean function \( f : \{0,1\}^n \to \{0,1\} \), our main result shows that
the bounded-error quantum and deterministic communication complexities of the
$\AND$-function \( f \circ \AND_2 \) are polynomially related, up to
polylogarithmic factors in \( n \). That is,
\[
\Dcc(f \circ \AND_2)
= \Qcc(f \circ \AND_2)^{O(1)} \cdot (\log n)^{O(1)} .
\]

To establish such a relationship, it is helpful to characterize
\( \Dcc(f \circ \AND_2) \) in terms of structural properties of the outer function
\( f \).
This was achieved by Knop, Lovett, McGuire, and Yuan~\cite{knop2021log}, who showed
that the deterministic communication complexity of \( f \circ \AND_2 \) is
characterized—up to polynomial loss and ignoring poly-log$(n)$ factors—by the logarithm of the \emph{De Morgan sparsity} of \( f \).

Recall that every Boolean function \( f \) admits a unique multilinear polynomial
representation over the reals,
\[
f(x) = \sum_{S \subseteq [n]} a_S \prod_{i \in S} x_i .
\]
The (De Morgan) \emph{sparsity} of \( f \), denoted \( \spar(f) \), is the number of
nonzero coefficients \( a_S \).

Knop et al.~\cite{knop2021log} showed that the deterministic communication
complexity of an $\AND$-function satisfies
\[
\Dcc(f \circ \AND_2)
= O\!\left( (\log \spar(f))^5 \cdot \log n \right).
\]

This characterization has an immediate and important consequence.
Since the sparsity of \( f \) coincides with the rank of the communication
matrix of \(f \circ \AND_2 \), this bound yields a resolution of the log-rank
conjecture for $\AND$-functions, up to a \( \log n \) factor.

Given this characterization of \( \Dcc(f \circ \AND_2) \), our task reduces to
showing that large sparsity of \( f \) forces large bounded-error quantum
communication complexity. Concretely, we show that
\[
\Qcc(f \circ \AND_2)
= (\log \spar(f))^{\Omega(1)},
\]
ignoring polylogarithmic factors in \( n \).
In fact, we prove a stronger statement: large sparsity of \( f \) implies a large
\emph{approximate} \( \gamma_2 \) norm of the communication matrix
\( M_{f \circ \AND_2} \).
This is a strengthening of a direct lower bound on quantum communication
complexity, since the approximate \( \gamma_2 \) norm is a known lower bound on
bounded-error quantum communication complexity.

The approximate \( \gamma_2 \) norm of a two-party Boolean function
\( F : X \times Y \to \{0,1\} \), denoted \( \agamma(F) \), is defined as the minimum
total weight of a rectangle decomposition that approximates the communication
matrix \( M_F \) within constant error:
\[
\agamma(F)
=
\min \Bigl\{
    \sum_i |\alpha_i| \;\Big|\;
    \bigl\| M_F - \sum_i \alpha_i R_i \bigr\|_\infty \le 1/3
\Bigr\},
\]
where each \( R_i(x,y) = g_i(x)h_i(y) \) is a combinatorial rectangle.

Any bounded-error quantum protocol for \( F \) using \( c \) qubits of
communication induces a pointwise approximation of \( M_F \) of the form
\[
M_F \approx \sum_i \alpha_i R_i ,
\]
with total weight \( \sum_i |\alpha_i| \le 2^{O(c)} \). While such decompositions are immediate for deterministic and randomized
protocols, they also hold in the quantum setting by unpacking the definition of
quantum communication protocols (see~\cite{linial2007lower}).

Our main technical contribution is to show that large sparsity of \( f \) lifts to
a lower bound on the approximate \( \gamma_2 \) norm of \( f \circ \AND_2 \).

\begin{restatable}{theorem}
{spargammatwo}\label{thm:spar-gamma_2}

For every total Boolean function \( f : \{0,1\}^n \to \{0,1\} \),
\[
\log \agamma(f \circ \AND_2)
=
\Omega\!\left(
    \left(\frac{\log \spar(f)}{\log n}\right)^{1/3}
\right).
\]
\end{restatable}
We now record several immediate consequences of
\cref{thm:spar-gamma_2}.
Recall that for a communication problem \(F : X \times Y \to \{0,1\}\), the
\emph{rank} of \(F\) is the rank (over the reals) of its communication matrix \(M_F\), and the
\emph{approximate rank}, denoted \(\arank(F)\), is the minimum rank of a real
matrix that approximates \(M_F\) entrywise within a small error 1/3.

The approximate \( \gamma_2 \) norm can be viewed as a convex relaxation of the
approximate rank and is known to be essentially equivalent to it.
It is shown in~\cite{lee2009approximation} that, on the logarithmic scale, the two
measures coincide up to an additive \( O(\log \log |X||Y|) \) term: for every
communication problem \( F : X \times Y \to \{0,1\} \),
\[
\Omega\!\bigl(\log \agamma(F)\bigr)
\;\le\;
\log \arank(F)
\;\le\;
O\!\bigl(\log \agamma(F) + \log \log |X||Y|\bigr).
\]
Combining this observation with our result and the result of Knop et
al.~\cite{knop2021log}, we obtain
\[
\begin{aligned}
\Rcc(f \circ \AND_2) \leq \Dcc(f \circ \AND_2)
&=
O\!\left((\log \spar(f))^5 \cdot \log n\right)\\
&=
O\!\left((\log \arank(f \circ \AND_2))^{15} \cdot (\log n)^6\right).
\end{aligned}
\]
While the exponents above are not optimized and can be improved (see \cref{subsec:consequences}), this implies that the
\emph{log-approximate-rank conjecture} holds for $\AND$-functions, up to
poly-log$(n)$ factors.

The Log-Approximate-Rank Conjecture (LARC), a term first coined by Lee and Shraibman~\cite{lee2009lower}, is a natural
approximate analogue of the classical Log-Rank conjecture.
It asserts that, just as the logarithm of the rank of the communication matrix is
believed to characterize deterministic communication complexity up to polynomial
loss, the logarithm of the approximate rank should similarly characterize
randomized communication complexity up to polynomial loss. 

Somewhat surprisingly, the LARC was recently shown to be false by Chattopadhyay, Mande, and Sherif~\cite{CMS20}, by exhibiting an $\XOR$-function as a counterexample.
In contrast, our results show that LARC \emph{does} hold for
$\AND$-functions, up to polylogarithmic factors in $n$.

In fact, our results give the following chain of inequalities:

\[
\begin{aligned}
\Omega\!\bigl(\log \agamma(f \circ \AND_2)\bigr)
&\;\le\;
\Qcc(f \circ \AND_2)
\;\le\;
\Rcc(f \circ \AND_2)
\;\le\;
\Dcc(f \circ \AND_2),
\\
\Dcc(f \circ \AND_2)
&\;\le\;
(\log \spar(f))^{O(1)} (\log n)^{O(1)}
\;\le\;
\log \agamma(f \circ \AND_2)^{O(1)} (\log n)^{O(1)} .
\end{aligned}
\]

As a result, for an $\AND$-function \( F := f \circ \AND_2 \), up to polylogarithmic factors in \( n \), the quantities
\[
\Dcc(F),\;
\Rcc(F),\;
\Qcc(F),\;
\log \agamma(F),\;
\log \arank(F),\;
\log \rank(F),
\text{ and}
\log \spar(f)
\]

are all polynomially equivalent. In other words, for $\AND$-functions, diverse
notions of complexity—from communication-theoretic to algebraic/analytic—coincide.
Interestingly, this striking equivalence of various measures was suspected in the early work of
Buhrman and de~Wolf~\cite{BuhrmanW01} on $\AND$-functions more than two decades ago.
Our results finally confirm their intuition. 

\subsection{Proof overview}

As discussed above, our main technical contribution is
Theorem~\ref{thm:spar-gamma_2}, which lifts the sparsity of a Boolean function
\( f \) to a lower bound on the approximate \( \gamma_2 \)-norm of the lifted
function \( f \circ \AND_2 \).
We now give a high-level overview of the proof.

We begin with a concrete example that will later help illustrate the general
argument.
The canonical Boolean function with large sparsity is the \( n \)-bit \( \OR_n \)
function, whose sparsity equals \( 2^n - 1 \).
When composed with the \( \AND_2 \) gadget, the function
\( \OR_n \circ \AND_2 \) corresponds to the \emph{Set Intersection} problem, the
negation of the well-known Set Disjointness function. 

\medskip
\noindent
\textbf{Lower bound for Set Intersection.}
We illustrate our approach using the Set Intersection function
\[
(\OR_n \circ \AND_2)(x,y)
\;=\;
\OR_n(x_1 \wedge y_1,\dots,x_n \wedge y_n),
\]
where we write \( z_i := x_i \wedge y_i \).

The key structural property of \( \OR_n \) is that it remains an $\OR$ under every
restriction \( \rho_z \in \{0,*\}^n \): the restricted function retains full degree
on the free variables. This hardness under \(\{0,*\}\)-restrictions is precisely
what gives \( \OR_n \) its large sparsity. More importantly, as we show later, the existence of many such
\emph{max-degree} restrictions is a general consequence of large sparsity, allowing
the argument to extend beyond Set Intersection to arbitrary $\AND$-functions.

To lift this hardness to the communication setting, we map restrictions
on the \( z \)-variables to restrictions on the input pairs \( (x_i,y_i) \).
Given \( \rho_z \in \{0,*\}^n \), we define a lifted restriction
\( \rho \) by
\[
(\rho(x_i), \rho(y_i)) =
\begin{cases}
(\Delta,0), & \text{if } \rho_z(z_i)=0,\\
(*,1), & \text{if }\rho_z(z_i)=*,
\end{cases}
\]
so that \( (x_i \wedge y_i)|_\rho = \rho_z(z_i) \).
Here \( \Delta \) denotes a free but \emph{masked} variable: although syntactically
free, the restricted function does not depend on it.

These lifted restrictions have three crucial properties:
(i) all \( y \)-variables are fixed,
(ii) the restricted function \( (\OR_n \circ \AND_2)|_{\rho} \) computes an $\OR$ on exactly the \( * \)-variables,
and (iii) it is independent of the masked variables.
As a result, if \( \rho_z \) has exactly \( d \) stars, then
\[
\deg\!\left(
\mathbb{E}_{x_{M(\rho)}}[(\OR_n \circ \AND_2)|_\rho]
\right)
= d,
\]
where \( M(\rho) \) denotes the masked variables and
\( \mathbb{E}_{x_{M(\rho)}}[\cdot] \) denotes expectation over independent uniform
assignments to them.

This motivates the following \emph{restriction-and-averaging procedure}.
We sample an unlifted restriction
\( \rho_z \in \{0,*\}^n \) uniformly at random subject to having
\( d = \Theta(n^{2/3}) \) stars (we will explain shortly why this specific parameter is chosen), lift it to obtain \( \rho \), and then take
expectation over the masked variables.

Under this restriction-and-averaging procedure,
\( \OR_n \circ \AND_2 \) retains its hardness—in terms of degree—with exact degree
\( \Theta(n^{2/3}) \).
On the other hand, we show that the same procedure simplifies any small-weight
rectangle decomposition arising from a small approximate \( \gamma_2 \) norm.
Concretely, it converts such a decomposition into a low-degree polynomial (more
precisely, one whose total Fourier mass on high-degree monomials is small).

This leads to a contradiction.
Indeed, from a small approximate \( \gamma_2 \) norm of size
\(2^{O(n^{1/3})}\) for \( \OR_n \circ \AND_2 \), we would obtain a polynomial
of degree \(O(n^{1/3})\) that approximates
\(
\mathbb{E}_{x_{M(\rho)}}[(\OR_n \circ \AND_2)|_{\rho}]
\),
even though the latter has exact degree \( \Theta(n^{2/3}) \). This contradicts the known quadratic relationship between degree and approximate degree,
forcing the approximate \( \gamma_2 \) norm of
\( \OR_n \circ \AND_2 \) to be \( 2^{\Omega(n^{1/3})} \). Extracting low-degree approximating polynomials from a decomposition exhibiting low approximate \(\gamma_2 \) norm of a matrix, in the context of communication complexity, was initiated in the work of Razborov \cite{razborov2003quantum}, followed by other works, including that of \cite{BuhrmanVW07,Sherstov16}. Our particular method is inspired by the work of \cite{chattopadhyay2026restriction} and the work of Krause and Pudl{\'a}k \cite{KP95} who used it in the context of proving lower bounds on the Fourier sparsity of lifted functions.

We now briefly explain why our restriction-and-averaging procedure simplifies a
small-weight rectangle decomposition into a low-degree polynomial.
Let
\[
\Pi(x,y) \;=\; \sum_{i=1}^m b_i\, g_i(x)\, h_i(y)
\]
be an approximator for Set Intersection with small total weight
\( \sum_i |b_i| \).
Under any lifted restriction, all \( y \)-variables are fixed, so each rectangle
collapses to a function of the \( x \)-variables alone.
Moreover, for any Boolean function \( g : \{0,1\}^n \to \{0,1\} \) on $x$-variables, we show that the
expected Fourier mass of \( g \) above level \( k \geq n^{1/3} \) after restriction and
averaging decays as \( 2^{-\Omega(k)} \) (this is where our choice of $d=\Theta(n^{2/3})$ plays a part).
Intuitively, this happens because any Fourier character involving a masked
variable vanishes upon averaging. An averaging argument then yields a restriction \( \rho \) such that the Fourier
mass of
\(
\mathbb{E}_{x_{M(\rho)}}[\Pi|_{\rho}]
\)
above level \( k = O(n^{1/3}) \) is negligible, provided
\( \sum_i |b_i| \le 2^{O(n^{1/3})} \).
Discarding higher-degree monomials produces a low-degree $O(n^{1/3})$ approximator for
\(
\mathbb{E}_{x_{M(\rho)}}[(\OR_n \circ \AND_2)|_{\rho}],
\)
contradicting its degree \( \Theta(n^{2/3}) \).

This shows a \( 2^{\Omega(n^{1/3})} \) lower bound on the approximate
\( \gamma_2 \) norm of Set Intersection.
Moreover, if you observe closely, the same argument also yields a
\( 2^{\Omega(n^{1/3})} \) lower bound for the approximate
\( \gamma_2 \) norm of Inner Product function
\( \XOR_n \circ \AND_2 \), another central problem in communication complexity.

\medskip
\noindent
\textbf{Generalization to arbitrary functions with large sparsity.}  In the discussion above for Set Intersection and Inner Product, we relied crucially on a structural property of the outer functions \( \OR_n \) and \( \XOR_n \): under every restriction in \( \{0,*\}^n \), the restricted function retains full degree on the surviving free variables. As hinted earlier, the existence of many such \emph{max-degree} restrictions is not specific to these functions. Rather, it is a general consequence of large sparsity—and this is the key insight of our work
that allows us to extend the argument to arbitrary $\AND$-functions.

More concretely, we show that any Boolean function
\( f : \{0,1\}^n \to \{0,1\} \) of large sparsity admits a fixed set of variables \( V \subseteq [n] \), which we call the \emph{core variables}, with the following property. For every assignment
\( \alpha \in \{0,*\}^V \),
the remaining variables in \( [n] \setminus V \) can be fixed—possibly depending on \( \alpha \)—so that the resulting restriction of \( f \) has full degree on the free variables.
We refer to the resulting collection of restrictions as a
\emph{semi-adaptive max-degree restriction tree}.
For \( \OR_n \) and \( \XOR_n \), the core set is simply \( V = [n] \), but in
general this need not be the case.
For a function of sparsity \( s \), we show that one can always find V of size
\(
\Omega(\log s/\log n),
\)
which we call the \emph{depth} of the restriction tree.

Having obtained such a collection of restrictions \( \mathcal{D} \) for \( f \),
we lift them to restrictions for the two-party function
\( f \circ \AND_2 \), extending the construction used for Set Intersection.
Given a restriction \( \rho_z \in \mathcal{D} \), we define a lifted restriction
\( \rho \) on the variables \( (x_i,y_i) \) by
\[
(\rho(x_i), \rho(y_i)) =
\begin{cases}
(\Delta,0), & \text{if } \rho_z(z_i)=0 \text{ and } i \in V,\\
(*,1), & \text{if } \rho_z(z_i)=* \text{ and } i \in V,\\
(1,1), & \text{if } \rho_z(z_i)=1 \text{ and } i \notin V,\\
(1,0), & \text{if } \rho_z(z_i)=0 \text{ and } i \notin V,
\end{cases}
\]
ensuring that \( (x_i \wedge y_i)|_\rho = \rho_z(z_i) \).
As in the Set Intersection case, this construction guarantees that \( (f \circ \AND_2)|_{\rho} \) coincides with \( f|_{\rho_z} \) (up to renaming of
variables) and is independent of the masked variables \( M(\rho) \).
As a result, for every max-degree restriction \( \rho_z \),
\[
\deg\!\left(
\mathbb{E}_{x_{M(\rho)}}\bigl[(f \circ \AND_2)|_{\rho}\bigr]
\right)
= \deg\!\left(
f|_{\rho_z}
\right)
=|\rho_z^{-1}(*)|
=|\rho^{-1}(*)|.
\]

Thus, the lifted function \( f \circ \AND_2 \) retains its hardness under the same
restriction-and-averaging procedure used for Set Intersection.
Moreover, the lifting has an additional convenient feature:
all non-core \( x \)-variables are fixed to \( 1 \).
This allows the same Fourier-analytic argument to go through, showing that the
restriction-and-averaging procedure simplifies any small-weight rectangle
decomposition into a low-degree polynomial.

The only quantitative change from the Set Intersection analysis is that the
effective parameter \( n \) is replaced by the number of core variables
\( |V| = \Omega(\log s / \log n) \).
Carrying out the argument yields
\[
\log \agamma(f \circ \AND_2)
=
\Omega\!\left(
    \left(\frac{\log s}{\log n}\right)^{1/3}
\right),
\]
which is exactly the bound stated in
Theorem~\ref{thm:spar-gamma_2}.

Finally, we place our work in the context of the results of
Chattopadhyay, Dahiya, and Lovett~\cite{chattopadhyay2026restriction}.
Among other results, that work studies the sparsity of Boolean functions and
provides a structural characterization in terms of \emph{max-degree restriction
trees}.
Their restriction trees are fully adaptive: the choice of which variable to
restrict next depends on the outcomes of previous restrictions, and there is no
fixed set of core variables, as in our setting.
Moreover, their framework applies to a broader class of complexity measures,
which they call one-sided and two-sided nice measures.

While this adaptive viewpoint is powerful, we do not know how to use such
fully adaptive restriction trees to control analytic quantities such as the
approximate \( \gamma_2 \) norm.
In contrast, our work constructs a more structured, \emph{semi-adaptive} form of
max-degree restriction trees directly from sparsity, with a fixed core set of
variables on which all \(\{0,*\}\)-restrictions occur.
This additional structure is crucial for our analysis and enables us to lift
sparsity to lower bounds on the approximate \( \gamma_2 \) norm.

\paragraph*{Organization}
In~\cref{sec:preliminaries}, we introduce the necessary preliminaries and notation.
In~\cref{sec:rest-tree-construction}, we formally define semi-adaptive
max-degree restriction trees and show how to construct them for Boolean functions
with large sparsity.
Finally, in~\cref{sec:lifting}, we prove our main technical contribution: a lifting
theorem that lifts the sparsity of a Boolean function \( f \) into a lower
bound on the approximate \( \gamma_2 \) norm of its lifted function
\( f \circ \AND_2 \), and we discuss the resulting consequences for
$\AND$-functions.

\section{Preliminaries}
\label{sec:preliminaries}

All functions considered are defined on the Boolean hypercube
\( \{0,1\}^n \) and unless stated otherwise, all polynomials are real and multilinear.

\paragraph*{Multilinear representations}
Over the Boolean domain, Boolean functions admit a canonical polynomial
representation.

\begin{definition}[Multilinear polynomial representation]
A polynomial \( Q \in \mathbb{R}[x_1,\dots,x_n] \) is \emph{multilinear} if each
variable appears with degree at most one in every monomial.
Every function \( f : \{0,1\}^n \to \mathbb{R} \) admits a unique multilinear
polynomial representation; that is, there exists a unique multilinear polynomial
\( Q \in \mathbb{R}[x_1,\dots,x_n] \) such that
\( Q(x) = f(x) \) for all \( x \in \{0,1\}^n \).
\end{definition}

We freely identify Boolean functions with their unique multilinear polynomial
representations, and use the two interchangeably.

\paragraph*{Polynomial complexity measures}
Let \( Q \in \mathbb{R}[x_1,\dots,x_n] \) be a multilinear polynomial written as
\[
    Q(x) = \sum_{S \subseteq [n]} a_S \prod_{i \in S} x_i .
\]
The \emph{degree} of \( Q \) is
\(
\deg(Q) := \max\{ |S| : a_S \neq 0 \},
\)
and the \emph{sparsity} of \( Q \), denoted \( \spar(Q) \), is the number of
nonzero coefficients \( a_S \).

For a function \( f : \{0,1\}^n \to \mathbb{R} \), let \( \mathcal{P}(f) \) denote
its unique multilinear polynomial representation. We define
\[
    \deg(f) := \deg(\mathcal{P}(f)), \qquad
    \spar(f) := \spar(\mathcal{P}(f)).
\]

\paragraph*{Approximate degree}
Let \( f : \{0,1\}^n \to \mathbb{R} \) and \( \epsilon > 0 \).
The \emph{\(\epsilon\)-approximate degree} of \( f \) is defined as
\[
\adeg_\epsilon(f)
:= \min \bigl\{
    \deg(Q) :
    |Q(x) - f(x)| \le \epsilon \text{ for all } x \in \{0,1\}^n
\bigr\}.
\]
When \( \epsilon = 1/3 \), we write \( \adeg(f) := \adeg_{1/3}(f) \).

\begin{theorem}[{\cite[Theorem~10, Section~3.4]{BunT22}}]
\label{thm:error-reduction}
Let \( f : \{0,1\}^n \to \{0,1\} \) be a Boolean function. For any
\( 0 < \epsilon < 1/2 \),
\[
    \adeg_\epsilon(f)
    = O\bigl(\adeg(f) \cdot \log(1/\epsilon)\bigr).
\]
\end{theorem}

\begin{theorem}[{\cite[Theorem~4]{aaronson2021degree}}]
\label{thm:degree-apxdegree}
For every Boolean function \( f : \{0,1\}^n \to \{0,1\} \),
\[
    \deg(f) = O\bigl(\adeg(f)^2\bigr).
\]
\end{theorem}

\paragraph*{Fourier basis}
Another fundamental representation of Boolean functions is given by the
\emph{Fourier basis}, where each monomial $\chi_S(x) = (-1)^{\sum_{i \in S} x_i}$ represents the $\pm1$-valued parity function on the subset $S \subseteq [n]$ of variables.

\begin{definition}[Fourier complexity measures]
Let \( f : \{0,1\}^n \to \mathbb{R} \) have the Fourier expansion
\[
  f(x) = \sum_{S \subseteq [n]} \widehat{f}(S)\, \chi_S(x).
\]
The \emph{Fourier degree} of \( f \) is the maximum size of a set
\( S \subseteq [n] \) with \( \widehat{f}(S) \neq 0 \).
Since
\( \chi_S(x) = \prod_{i \in S} (1 - 2x_i) \),
the Fourier degree coincides with the ordinary degree.

The \emph{Fourier sparsity} of \( f \), denoted \( \fsparsity{f} \), is the number
of nonzero Fourier coefficients \( \widehat{f}(S) \).
The \emph{Fourier \( \ell_1 \)-norm} of \( f \) is
\[
  \fourierwt{f} := \sum_{S \subseteq [n]} |\widehat{f}(S)|.
\]
For \( t \ge 0 \), the \emph{Fourier \( \ell_1 \)-mass above level \( t \)} is
\[
  \fourierwt{f}^{\ge t}
  := \sum_{\substack{S \subseteq [n]\\ |S| \ge t}} |\widehat{f}(S)|.
\]
\end{definition}

\paragraph*{Communication complexity}
We assume familiarity with the standard model of communication complexity and
refer the reader to \cite{KN-CC-book} for background.
In this model, two parties—Alice and Bob—aim to compute a Boolean function
\( F : X \times Y \to \{0,1\} \), where Alice receives \( x \in X \) and Bob
receives \( y \in Y \).
The \emph{deterministic communication complexity} of \( F \), denoted
\( \Dcc(F) \), is the minimum number of bits exchanged by any deterministic
protocol that always outputs \( F(x,y) \).
In the public-coin randomized model, Alice and Bob have access to shared
randomness and must compute \( F(x,y) \) with error at most \( 1/3 \); the
corresponding measure is the \emph{randomized communication complexity}
\( \Rcc(F) \).

We also assume familiarity with quantum communication complexity
\cite{de2002quantum}.
We use \( \Qcc(F) \) to denote the bounded-error (error at most \( 1/3 \)) quantum
communication complexity of \( F \) in the model with unlimited shared
entanglement.

\paragraph*{The $\gamma_2$ and approximate $\gamma_2$ norms}

\begin{definition}[$\gamma_2$ and approximate $\gamma_2$ norms]
Let \( F : X \times Y \to \{0,1\} \) be a Boolean function, and let
\( M_F \in \{0,1\}^{X \times Y} \) denote its communication matrix, defined by
\( M_F(x,y) = F(x,y) \).

The \emph{$\gamma_2$ norm} of \( F \) is
\[
\gamma_2(F)
=
\min \Bigl\{
    \sum_i |\alpha_i| :
    M_F = \sum_i \alpha_i R_i
\Bigr\},
\]
where each \( R_i \) is a \emph{combinatorial rectangle}, i.e.,
\( R_i(x,y) = g_i(x) h_i(y) \) for Boolean functions
\( g_i : X \to \{0,1\} \) and \( h_i : Y \to \{0,1\} \).

For \( 0 < \epsilon < 1/2 \), the \emph{\(\epsilon\)-approximate $\gamma_2$ norm}
is
\[
\gamma_2^\epsilon(F)
=
\min \Bigl\{
    \gamma_2(A) :
    A \in \mathbb{R}^{X \times Y},\;
    \|A - M_F\|_\infty \le \epsilon
\Bigr\}.
\]
We write \( \agamma(F) := \gamma_2^{1/3}(F) \).
\end{definition}

\begin{remark}
The $\gamma_2$ norm is often defined via a factorization-based formulation.
For a real matrix \( M \), one may equivalently define
\[
\gamma_2(M)
=
\min_{X,Y : XY^\top = M} r(X)\, r(Y),
\]
where \( r(X) \) denotes the maximum \( \ell_2 \)-norm of a row of \( X \).
This formulation, sometimes called the \emph{$\mu$-norm}, is equivalent to the
rectangle-based definition up to constant factors; see, e.g.,
Chapter~2 of \cite{lee2009lower}.
We use the rectangle-based definition throughout.
\end{remark}

\begin{theorem}[{\cite[Theorem~1]{linial2007lower}}]
\label{thm:agamma-qcc}
For every Boolean function \( F : X \times Y \to \{0,1\} \),
\[
\Rcc(F) = \Omega(\log \agamma(F))
\qquad\text{and}\qquad
\Qcc(F) = \Omega(\log \agamma(F)).
\]
\end{theorem}

\section{Semi-Adaptive Max-Degree Restriction Trees from Sparsity}
\label{sec:rest-tree-construction}

\begin{definition}[Restrictions]
A \emph{restriction} \( \rho \) on a set of variables \( V \subseteq \{x_1, \dots, x_n\} \) is a partial assignment
\(
    \rho : V \to \{0,1,*\},
\)
where for \( x_i \in V \), \( \rho(x_i) \in \{0,1\} \) indicates that \( x_i \) is fixed, and \( \rho(x_i) = * \) means \( x_i \) is left free. 
For a polynomial \( Q \in \mathbb{R}[x_1, \dots, x_n] \), we write \( Q|_\rho \) for the polynomial obtained by substituting \( x_i = \rho(x_i) \) for all fixed variables \( x_i \).
\end{definition}

A central tool in our work is to extract structural consequences of a function
having large polynomial sparsity. In particular, we seek restrictions under which
a function remains maximally hard, in the sense of retaining full degree.

At a high level, large sparsity guarantees the existence of many restrictions
under which the function retains full degree. More concretely, if a multilinear
polynomial has sparsity \( s \), then there exists a set of variables \( V \) of
size \( \Omega(\log s / \log n) \) such that, for every assignment of the variables
in \( V \) to values in \( \{0,*\} \), one can fix the remaining variables so that
the restricted polynomial has full degree in the variables left free.

We capture this collection of restrictions using what we call
\emph{semi-adaptive max-degree restriction trees}.
The term \emph{semi-adaptive} reflects the following structure.
There is a fixed set of variables \( V \) such that every assignment in
\( \{0,*\}^V \) appears as a restriction, rather than variables being chosen
adaptively based on previous assignments, as in fully adaptive restriction trees
(e.g., in the work of~\cite{chattopadhyay2026restriction}).
However, the restrictions are not fully non-adaptive: although the set \( V \) is
fixed, for each assignment \( \rho \in \{0,*\}^V \), the fixing of the remaining
variables in \( [n]\setminus V \) that ensures full degree may depend on
\( \rho \). This intermediate structure motivates the term \emph{semi-adaptive}.

The qualifier \emph{max-degree} indicates that under every restriction in the
tree, the polynomial retains full degree on the variables that remain free.
This notion compactly encodes the key structural consequence of large sparsity
that we exploit later.

We now formalize this notion.

\begin{definition}[Max-degree restriction]
Let $Q \in \mathbb{R}[x_1,\dots,x_n]$ be a nonzero multilinear polynomial and let
$\rho : \{x_1,\dots,x_n\} \to \{0,1,*\}$ be a restriction.
We say that $\rho$ is a \emph{max-degree restriction} of $Q$ if
\(
\deg(Q|_\rho) = |\rho^{-1}(*)|,
\)
that is, the restricted polynomial has full degree in its free variables.
\end{definition}

\begin{definition}[Semi-adaptive restriction tree]
A \emph{semi-adaptive restriction tree of depth $d$} on $n$ variables
is a collection $\mathcal{D}$ of $2^d$ restrictions
$\rho : \{x_1,\dots,x_n\} \to \{0,1,*\}$
for which there exists a fixed set of variables
$V \subseteq \{x_1,\dots,x_n\}$ with $|V| = d$,
called the \emph{core variables}, such that:
\begin{itemize}
    \item For every $\rho \in \mathcal{D}$ and every $x_i \notin V$,
    $\rho(x_i) \in \{0,1\}$.
    \item For every assignment $\alpha \in \{0,*\}^V$, there exists a unique
    $\rho \in \mathcal{D}$ such that $\rho(x_i) = \alpha(x_i)$ for all $x_i \in V$.
\end{itemize}
Equivalently, $\mathcal{D}$ consists of all restrictions obtained by assigning each
variable in $V$ either $0$ or $*$, while fixing all variables outside $V$ as a
function of this assignment.
\end{definition}

\begin{definition}[Semi-adaptive max-degree restriction tree]
Let $Q \in \mathbb{R}[x_1,\dots,x_n]$ be a multilinear polynomial and let
$\mathcal{D}$ be a semi-adaptive restriction tree.
We say that $\mathcal{D}$ is a \emph{max-degree restriction tree} for $Q$ if every
$\rho \in \mathcal{D}$ is a max-degree restriction of $Q$, i.e.,
\(
\deg(Q|_\rho) = |\rho^{-1}(*)|.
\)
\end{definition}

\paragraph*{Examples}
To illustrate the definitions, we give two informative examples.

\begin{example}
\( \OR_n \).
The function \( \OR_n(x_1,\dots,x_n) \) has sparsity \( 2^n-1 \).
It admits a semi-adaptive max-degree restriction tree of depth \( n \) with core
variables \( V=\{x_1,\dots,x_n\} \).
Let \( \mathcal{D}=\{\,\rho:\{x_1,\dots,x_n\}\to\{0,*\}\,\} \).
For any \( \rho\in\mathcal{D} \), the restricted function \( \OR_n|_{\rho} \) is
an $\OR$ over the free variables \( \rho^{-1}(*) \), and hence
\( \deg(\OR_n|_{\rho})=|\rho^{-1}(*)| \).
Thus \( \mathcal{D} \) is a semi-adaptive max-degree restriction tree for \( \OR_n \).
\end{example}

\begin{example} \( \AND_n\circ\OR_2 \).
Consider
\[
(\AND_n\circ\OR_2)(x_1,\dots,x_n,y_1,\dots,y_n)
= \AND_n(\OR_2(x_1,y_1),\dots,\OR_2(x_n,y_n)),
\]
which has sparsity \( 3^n \).
We construct a semi-adaptive max-degree restriction tree of depth \( n \) with
core variables
\(
V = \{x_1,\dots,x_n\}.
\)
For a restriction \( \rho_x \in \{0,*\}^{\{x_1,\dots,x_n\}} \), define a
restriction \( \rho_y^{\rho_x} \) on \( \{y_1,\dots,y_n\} \) by
\[
\rho_y^{\rho_x}(y_i) \;=\;
\begin{cases}
0, & \text{if } \rho_x(x_i) = *, \\
1, & \text{if } \rho_x(x_i) = 0.
\end{cases}
\]
Set
\(
\mathcal{D}
\;=\;
\bigl\{\, \rho_x \cup \rho_y^{\rho_x} \;\big|\;
\rho_x \in \{0,*\}^{\{x_1,\dots,x_n\}} \,\bigr\}.
\)
For any \( \rho \in \mathcal{D} \), each gate \( \OR_2(x_i,y_i) \) evaluates to
a free variable when \( \rho(x_i)=* \), and to the constant \( 1 \) when
\( \rho(x_i)=0 \).
As a result, \( (\AND_n \circ \OR_2)|_{\rho} \) computes an $\AND$ over exactly the
free core variables \( \rho^{-1}(*) \), and therefore
\(
\deg\bigl((\AND_n \circ \OR_2)|_{\rho}\bigr)
= |\rho^{-1}(*)|.
\)
Thus, \( \mathcal{D} \) is a semi-adaptive max-degree restriction tree of depth
\( n \) for \( \AND_n \circ \OR_2 \).
\end{example}

\begin{figure}[ht]
\centering
\begin{tikzpicture}[
  level distance=15mm,
  level 1/.style={sibling distance=60mm},
  level 2/.style={sibling distance=30mm},
  level 3/.style={sibling distance=15mm},
  every node/.style={font=\small, inner sep=2pt},
  edge from parent/.style={draw, -{Latex[length=2mm,width=1.5mm]}, thin},
  leaf/.style={rectangle, draw=black!60, rounded corners, fill=black!5, inner sep=4pt}
]

\node (root) [circle, draw] {$x_1$}
  child { 
    node [circle, draw] {$x_2$}
      child {
        node [circle, draw] {$x_3$}
         child { node [leaf, name=L11] {\shortstack{\((0,0,0)\)\\\((1,1,1)\)}}
         edge from parent node[left] {$0$ } }
          child { node [leaf] {\shortstack{\((0,0,*)\)\\\((1,1,0)\)}} edge from parent node[right] {$*$} }
        edge from parent node[left] {$0$}
      }
      child {
        node [circle, draw] {$x_3$}
          child { node [leaf] {\shortstack{\((0,*,0)\)\\\((1,0,1)\)}} edge from parent node[left] {$0$} }
          child { node [leaf] {\shortstack{\((0,*,*)\)\\\((1,0,0)\)}}
          edge from parent node[right] {$*$} }
        edge from parent node[right] {$*$}
      }
    edge from parent node[left] {$0$}
  }
  child { 
    node [circle, draw] {$x_2$}
      child {
        node [circle, draw] {$x_3$}
          child { node [leaf] {\shortstack{\((*,0,0)\)\\\((0,1,1)\)}} edge from parent node[left] {$0$} }
          child { node [leaf] {\shortstack{\((*,0,*)\)\\\((0,1,0)\)}} edge from parent node[right] {$*$} }
        edge from parent node[left] {$0$}
      }
      child {
        node [circle, draw] {$x_3$}
          child { node [leaf] {\shortstack{\((*,*,0)\)\\\((0,0,1)\)}} edge from parent node[left] {$0$} }
          child { node [leaf] {\shortstack{\((*,*,*)\)\\\((0,0,0)\)}}
          edge from parent node[right] {$*$} }
        edge from parent node[right] {$*$}
      }
    edge from parent node[right] {$*$}
  };

\node[font=\small, align=right, anchor=east] at ([xshift=-8mm]L11.center) {x:\\[3pt]y:};
\end{tikzpicture}
\caption{
A semi-adaptive max-degree restriction tree for \( \AND_3 \circ \OR_2 \) with core variables \( V = \{x_1,x_2,x_3\} \). Leaves correspond to the \(2^{|V|}\) restrictions in the tree. For each leaf restriction, the restricted function \( (\AND_3 \circ \OR_2)|_\rho \) computes an \(\AND\) over exactly the variables left free.}
\label{fig:or3-rest-tree}
\end{figure}
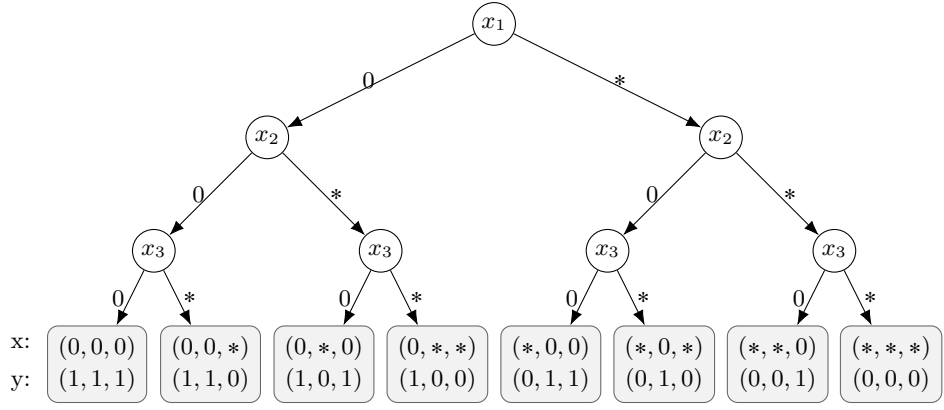

The following lemma shows that large sparsity guarantees the existence of deep
semi-adaptive max-degree restriction trees.

\begin{lemma}\label{lemma:spars-nonadapt-res-tree}
Let $Q : \{0,1\}^n \to \mathbb{R}$ be a nonzero multilinear polynomial of sparsity $s$.
Then $Q$ admits a semi-adaptive max-degree restriction tree of depth
$\Omega(\log s / \log n)$.
\end{lemma}

\begin{proof}
Let $\mathcal{M}_Q \subseteq 2^{\{x_1,\dots,x_n\}}$ denote the family of supports of
monomials appearing in $Q$. Let $V \subseteq \{x_1,\dots,x_n\}$ be a largest set
shattered by $\mathcal{M}_Q$. Recall that $V$ being shattered by $\mathcal{M}_Q$ means that for every $S \subseteq V$, there exists $T \in \mathcal{M}_Q$ such that $T \cap V = S$. The Sauer--Shelah--Perles lemma \cite{sauer1972density, shelah1972combinatorial} states that that for a set system $\mathcal{F}$ containing subsets of $[n]$, if $d$ is the maximum size of a set shattered by $\mathcal{F}$, then $|\mathcal{F}| \leq O(n^d)$. Using this lemma, we obtain
$|V| = \Omega(\log s / \log n)$; write $|V| = d$.

For each subset $S \subseteq V$, we construct a restriction $\rho_S$.
The resulting family $\{\rho_S : S \subseteq V\}$ will form a semi-adaptive
max-degree restriction tree for $Q$.

Fix an arbitrary $S \subseteq V$. Define the \emph{free restriction}
$\rho_S^{\mathrm{free}} : V \to \{0,*\}$ by setting
$\rho_S^{\mathrm{free}}(x)=*$ for all $x \in S$ and
$\rho_S^{\mathrm{free}}(x)=0$ for all $x \in V \setminus S$.
Under this restriction, every monomial of $Q$ containing a variable from
$V \setminus S$ vanishes. As a result, the restricted polynomial can be written as
\(
Q|_{\rho_S^{\mathrm{free}}}
= \sum_{T \subseteq S} \Bigl(\prod_{x \in T} x\Bigr)\cdot R_T,
\)
where each $R_T$ is a multilinear polynomial over the variables
$\{x_i : x_i \notin V\}$.

Since $V$ is shattered by $\mathcal{M}_Q$, there exists a monomial of $Q$ whose
support intersects $V$ exactly in $S$. Equivalently, the coefficient polynomial
$R_S$ is nonzero. We now fix the remaining variables to witness this nonzeroness.
Choose a \emph{fixing restriction}
$\rho_S^{\mathrm{fix}} : \{x_1,\dots,x_n\} \setminus V \to \{0,1\}$
such that $R_S|_{\rho_S^{\mathrm{fix}}} \neq 0$.

Let $\rho_S := \rho_S^{\mathrm{free}} \cup \rho_S^{\mathrm{fix}}$. Then
$Q|_{\rho_S}$ has degree exactly $|S|$, and its set of free variables is
$\rho_S^{-1}(*) = S$. Hence, $\rho_S$ is a max-degree restriction of $Q$.

Therefore, the family $\{\rho_S : S \subseteq V\}$ forms a semi-adaptive
max-degree restriction tree of depth $|V| = \Omega(\log s / \log n)$.
\end{proof}

\section{Lifting with the \texorpdfstring{$\AND_2$}{AND2} Gadget}
\label{sec:lifting}

In this section, we present our main technical contribution: a lifting theorem
that lifts the sparsity of a Boolean function
\( f : \{0,1\}^n \to \{0,1\} \)
into a lower bound on the approximate \( \gamma_2 \) norm of the lifted function
\( F := f \circ \AND_2 \).
We begin with a high-level overview of the proof.

\paragraph*{Proof overview}
Let \( f : \{0,1\}^n \to \{0,1\} \) be a Boolean function of sparsity \( s \).
Our goal is to lower bound the approximate \( \gamma_2 \) norm of the lifted function
\( f \circ \AND_2 \) in terms of \( \log s \).
Suppose, toward a contradiction, that \( f \circ \AND_2 \) admits a small approximate
\( \gamma_2 \) norm. Then there exists an approximator
\[
\Pi(x,y) = \sum_{i=1}^m b_i\, g_i(x) h_i(y)
\]
with \( \sum_i |b_i| \) small.
Our goal is to show that such an approximator cannot exist.
The proof proceeds via a carefully designed random restriction argument that preserves the hardness of \( f \circ \AND_2 \), in terms of degree, while simplifying any
such approximator.

\medskip
\noindent
\textbf{Step 1: Structure from sparsity.}
A key consequence of large sparsity is that
\( f \) admits a \emph{semi-adaptive max-degree restriction tree} \( \mathcal{D} \) of depth
\( d = \Omega(\log s / \log n) \).
Equivalently, there exists a fixed set of core variables
\( V \subseteq \{z_1,\dots,z_n\} \), with \( |V| = d \), such that
for every assignment in \( \{0,*\}^V \),
the remaining variables can be fixed so that the restricted function
has full degree in the surviving free variables.
This provides a large family of restrictions under which \( f \) remains maximally
hard in terms of degree.

\medskip
\noindent
\textbf{Step 2: Lifting restrictions through the \( \AND_2 \) gadget.}
We lift the restriction tree \( \mathcal{D} \) for \( f \) to a collection of
restrictions \( \mathcal{D} \circ \AND_2 \) for the lifted function
\[
(f \circ \AND_2)(x_1,\dots,x_n,y_1,\dots,y_n)
= f(\AND_2(x_1,y_1),\dots,\AND_2(x_n,y_n)).
\]
Each restriction \( \rho_f \in \mathcal{D} \) is mapped to a restriction
\( \rho \) on the variables \( (x_i,y_i) \) so that
\(
\AND_2(x_i,y_i)|_{\rho} = \rho_f(z_i)
\)
for every \( i \).
Concretely, the lifted restriction is given by
\[
(\rho(x_i), \rho(y_i)) =
\begin{cases}
(\Delta, 0), & \text{if } \rho_f(z_i) = 0 \text{ and } z_i \in V,\\
(*, 1), & \text{if } \rho_f(z_i) = * \text{ and } z_i \in V,\\
(1, 1), & \text{if } \rho_f(z_i) = 1 \text{ and } z_i \notin V,\\
(1, 0), & \text{if } \rho_f(z_i) = 0 \text{ and } z_i \notin V.
\end{cases}
\]
Here both \( * \) and \( \Delta \) denote free variables.
\cref{algo:fourier-restriction} formalizes this construction.

\medskip
The lifted restrictions satisfy two immediate properties.
First, all \( y \)-variables are fixed under every lifted restriction.
Second, the lifting introduces a special type of free variable, called a
\emph{masked variable}: these are free \( x_i \)-variables assigned \( \Delta \),
whose corresponding \( y_i \) is fixed to \( 0 \); consequently, the restricted
function is independent of them.

As a result, for every lifted restriction \( \rho \), averaging over the masked
variables yields a function that coincides (up to renaming variables) with
\( f|_{\rho_f} \).
For a restriction \( \rho \), we denote the set of masked variables by
\[
M(\rho) := \{\, x_i \mid \rho(x_i) = \Delta \,\},
\]
and write \( \mathbb{E}_{x_{M(\rho)}}[\cdot] \) for expectation over independent,
uniform assignments to the variables in \( M(\rho) \).
We will use this notation throughout.

Since \( \mathcal{D} \) is a max-degree restriction tree, it follows that
\[
\deg\!\left(\mathbb{E}_{x_{M(\rho)}}[(f \circ \AND_2)|_{\rho}]\right)
= |\rho^{-1}(*)|.
\]

\medskip
\noindent
\textbf{Step 3: Random restrictions and Fourier decay.}
We place a uniform distribution over lifted restrictions
\( \rho \in \mathcal{D} \circ \AND_2 \) having exactly
\( pd \) many \( * \)-variables.
We analyze the effect of sampling such a restriction and then taking expectation
over the masked variables on the following:
\begin{enumerate}
    \item \emph{The target function \( f \circ \AND_2 \).}
    This operation preserves hardness: by construction, the resulting function
    has exact degree \( pd \).

    \item \emph{An arbitrary Boolean function \( g \) on the \( x \)-variables.}
    We show that the same operation causes the Fourier
    \( \ell_1 \)-mass above level \( k \) to decay exponentially,
    provided \( k \gtrsim p^2d \).
\end{enumerate}

To derive a contradiction later, we need the Fourier mass above level
\( \sqrt{pd} \) to be exponentially small, since this would imply an
approximating polynomial of degree \( O(\sqrt{pd}) \) for the restricted target
function. Since exponential decay is guaranteed only above level
\( p^2d \), we choose parameters so that
\(
p^2d \approx \sqrt{pd}.
\)
Solving this relation yields
\(
p = \Theta(d^{-1/3}),
\)
for which
\(
pd = \Theta(d^{2/3}) \)
and \(
\sqrt{pd} = \Theta(d^{1/3}).
\)
Thus the Fourier tail decays exponentially above the precise level needed for
the contradiction.

\medskip
\noindent
\textbf{Step 4: Deriving a contradiction.}
Applying this random restriction-and-averaging procedure to the approximator
\( \Pi \), each function \( h_i(y) \) collapses to a constant,
while the Fourier tails of the corresponding \( g_i(x) \) terms decay rapidly.
Since \( \sum_i |b_i| \) is small, we conclude that the expected Fourier mass of
\(
\mathbb{E}_{x_{M(\rho)}}[\Pi|_{\rho}]
\)
above level
\(
k = \Theta(\sqrt{pd}) = \Theta(d^{1/3})
\)
is negligible.

Discarding this high-degree mass yields a polynomial of degree
\( O(\sqrt{pd}) \) that still approximates
\(
\mathbb{E}_{x_{M(\rho)}}[(f \circ \AND_2)|_{\rho}].
\)
However, by the construction above, this target function has exact degree
\( pd \), contradicting the general fact that the exact degree of a
Boolean function is at most quadratic in its approximate degree. Choosing parameters appropriately, this contradiction implies
\[
\log \agamma(f \circ \AND_2)
= \Omega\!\left(\left(\frac{\log s}{\log n}\right)^{1/3}\right),
\]
completing the proof.
\subsection{Lifting De Morgan Sparsity to the Approximate \texorpdfstring{$\gamma_2$}{gamma2}-Norm via the \texorpdfstring{$\AND_2$}{AND2} Gadget} \label{sec:lifted-restrictions}

We now formalize the proof strategy outlined in the proof overview.
The first step is to lift max-degree restrictions for \( f \)
to a collection of restrictions for the lifted function \( f \circ \AND_2 \).\cref{algo:fourier-restriction} describes this lifting procedure.

\begin{algorithm}[H]
\caption{\sampleliftedrho}
\label{algo:fourier-restriction}
\begin{algorithmic}[1]
\State \textbf{Input:} A semi-adaptive max-degree restriction tree
\( \mathcal{D} \) for \( f : \{0,1\}^n \to \{0,1\} \)
\State \textbf{Output:} A collection of lifted restrictions
\( \mathcal{D} \circ \AND_2 \), where each
\( \rho : \{x_i,y_i\}_{i=1}^n \to \{0,1,*,\Delta\} \)
is a restriction for \( f \circ \AND_2 \)

\State Let $V$ be the set of \emph{core variables} of
the semi-adaptive restriction tree \( \mathcal{D} \).

\For{each restriction $\rho_f \in \mathcal{D}$}
    \State Define the lifted restriction $\rho = \lift_{\mathcal{D}}(\rho_f)$ as follows:
    \For{each $i \in [n]$}
        \State Set
        \(
        (\rho(x_i), \rho(y_i)) =
        \begin{cases}
            (\Delta, 0), & \text{if } \rho_f(z_i) = 0 \text{ and } z_i \in V,\\
            (*, 1), & \text{if } \rho_f(z_i) = * \text{ and } z_i \in V,\\
            (1, 1), & \text{if } \rho_f(z_i) = 1 \text{ and } z_i \notin V,\\
            (1, 0), & \text{if } \rho_f(z_i) = 0 \text{ and } z_i \notin V.
        \end{cases}
        \)
    \EndFor
\EndFor

\State \Return \( \mathcal{D} \circ \AND_2 := \{\lift(\rho_f) : \rho_f \in \mathcal{D}\} \)
\end{algorithmic}
\end{algorithm}

A lifted restriction assigns each variable a value in
\( \{0,1,*,\Delta\} \).
When applying such a restriction to \( f \circ \AND_2 \),
both symbols \( * \) and \( \Delta \) are treated as free variables.
Variables marked by \( \Delta \) are called \emph{masked variables};
although syntactically free, the restricted function does not depend on them.
Tracking masked variables explicitly will be convenient for the subsequent analysis.

\paragraph*{Basic structure of lifted restrictions}
Let \( V \subseteq \{z_1,\dots,z_n\} \) denote the set of core variables of the
semi-adaptive restriction tree \( \mathcal{D} \), and let
\( V_x := \{ x_i : z_i \in V \} \) be the corresponding set of \( x \)-variables,
which we refer to as the \emph{core \(x\)-variables}.
We will use this notation throughout.

By construction, every lifted restriction
\( \rho \in \mathcal{D} \circ \AND_2 \) fixes all \( y \)-variables and fixes all
\( x \)-variables outside \( V_x \) to \(1\).
Moreover, the assignment on \( V_x \) uniquely determines the entire restriction:
for every \( \alpha \in \{\Delta,*\}^{V_x} \), there exists a unique
\( \rho \in \mathcal{D} \circ \AND_2 \) such that
\( \rho(x_i) = \alpha(x_i) \) for all \( x_i \in V_x \).
This follows directly from the defining property of \( \mathcal{D} \), which
guarantees that for every \( \beta \in \{0,*\}^{V} \) there is a unique
\( \rho_f \in \mathcal{D} \) satisfying \( \rho_f(z_i) = \beta(z_i) \) for all
\( z_i \in V \). As a result, the family \( \mathcal{D} \circ \AND_2 \) is  parametrized by
assignments to the core \( x \)-variables.

The construction of the lifted restrictions,
together with the fact that \( \mathcal{D} \) is a max-degree restriction tree,
implies the following structural properties. In particular, these properties show
that \( f \circ \AND_2 \) retains its hardness (in terms of degree) under the lifted restrictions.

\begin{claim}
\label{claim:liftedrho-properties}
Let \( \mathcal{D} \circ \AND_2 \) be the collection of lifted restrictions for
\( f \circ \AND_2 \) obtained from a semi-adaptive max-degree restriction tree
\( \mathcal{D} \) for \( f \) via \cref{algo:fourier-restriction}.
Then the following properties hold:
\begin{enumerate}
    \item For every \( \rho \in \mathcal{D} \circ \AND_2 \), the restricted function
    \( (f \circ \AND_2)|_{\rho} \) does not depend on the masked variables
    \( M(\rho) \).

    \item For every \( \rho \in \mathcal{D} \circ \AND_2 \),
    \( \deg\bigl(\mathbb{E}_{x_{M(\rho)}}[(f \circ \AND_2)|_{\rho}]\bigr) = \deg\bigl((f \circ \AND_2)|_{\rho}\bigr) = |\rho^{-1}(*)|. \)
\end{enumerate}
\end{claim}

\begin{proof}
For (1), if \(x_i\) is masked in \(\rho\), then by construction \(\rho(y_i)=0\).
Hence \(\AND_2(x_i,y_i)|_{\rho}=0\) regardless of the value of \(x_i\), and therefore
\((f \circ \AND_2)|_{\rho}\) is independent of all masked variables.

For (2), fix \(\rho \in \mathcal{D} \circ \AND_2\), and let \(\rho_f \in \mathcal{D}\)
be the restriction used to generate \(\rho\), i.e., \(\rho=\lift(\rho_f)\).
By construction, for every \(i \in [n]\) we have
\(\AND_2(x_i,y_i)|_{\rho}=\rho_f(z_i)\).
Thus, \((f \circ \AND_2)|_{\rho}\) depends on the variables \(\{x_i,y_i\}\) only
through the tuple \((\rho_f(z_1),\ldots,\rho_f(z_n))\).
After ignoring the masked variables (which \((f \circ \AND_2)|_{\rho}\) does not
depend on by part~(1)), the resulting function coincides with \(f|_{\rho_f}\).
Since \(\rho_f\) is a max-degree restriction of \(f\), we obtain
\(\deg((f \circ \AND_2)|_{\rho})=\deg(f|_{\rho_f})=|\rho_f^{-1}(*)|=|\rho^{-1}(*)|\).

Finally, since \( (f \circ \AND_2)|_{\rho} \) does not depend on the variables in
\( M(\rho) \), the function
\( \mathbb{E}_{x_{M(\rho)}}[(f \circ \AND_2)|_{\rho}] \)
is obtained by simply viewing \( (f \circ \AND_2)|_{\rho} \) as a function on the
remaining free variables. In particular, taking expectation over the masked variables
leaves the function unchanged as a polynomial in the remaining free variables.
Hence,
\(
\deg\bigl(\mathbb{E}_{x_{M(\rho)}}[(f \circ \AND_2)|_{\rho}]\bigr)
=
\deg\bigl((f \circ \AND_2)|_{\rho}\bigr).
\)
\end{proof}

Next, we introduce a probability distribution on the restrictions in
\( \mathcal{D} \circ \AND_2 \) for use in a random restriction argument.
Fix a parameter \( p \in (0,1) \), and recall that every restriction
\( \rho \in \mathcal{D} \circ \AND_2 \) satisfies
\( |\rho^{-1}(*)| + |\rho^{-1}(\Delta)| = d \), where \( d \) is the depth of the
semi-adaptive restriction tree \( \mathcal{D} \).

We consider the uniform distribution over all restrictions in
\( \mathcal{D} \circ \AND_2 \) that leave exactly \( pd \) variables 
marked by \( * \).
Formally, let
\(
U := \{\, \rho \in \mathcal{D} \circ \AND_2 \mid |\rho^{-1}(*)| = pd \,\}.
\) 
By construction, \( |U| = \binom{d}{pd} \).
We sample a restriction uniformly at random from \( U \), i.e., \( \Pr[\rho] = 1/\binom{d}{pd} \) for all \( \rho \in U \). We denote this distribution by \( \mathcal{U}_{p}(\mathcal{D} \circ \AND_2) \).

\medskip

Next, we analyze the effect of sampling a restriction
\( \rho \sim \mathcal{U}_{p}(\mathcal{D} \circ \AND_2) \), applying it to
\( f \circ \AND_2 \), and then taking expectation over the masked variables
\( M(\rho) \).
By \cref{claim:liftedrho-properties}, the resulting function retains degree \(pd\).
On the other hand, we show that applying the same process—sampling
\( \rho \sim \mathcal{U}_{p}(\mathcal{D} \circ \AND_2) \) and taking expectation over
\( M(\rho) \)—to an arbitrary Boolean function \(g\) over the
\(x\)-variables causes its high-degree Fourier mass to decay exponentially
above level \(k\), provided \(k \gtrsim p^2d\).
Optimizing this tradeoff leads to the choice \(p=\Theta(d^{-1/3})\), for which
\(pd=\Theta(d^{2/3})\) and the Fourier tail decays exponentially above level
\(k=\Theta(\sqrt{pd})=\Theta(d^{1/3})\).
Combining these two observations yields our main result: a lower bound on the
approximate \( \gamma_2 \) norm of \( f \circ \AND_2 \) in terms of the sparsity
of \( f \).

\begin{claim}\label{claim:prop2}
Let \( \mathcal{D} \circ \AND_2 \) be the collection of lifted restrictions for
\( f \circ \AND_2 \) obtained from a semi-adaptive max-degree restriction tree
\( \mathcal{D} \) of depth \( d \) via \cref{algo:fourier-restriction}.
Let \( g : \{0,1\}^n \to \{0,1\} \) be an arbitrary Boolean function on the
\( x \)-variables \( \{x_1,\dots,x_n\} \). Then for every integer \( k \ge 0 \),
\[
\mathbb{E}_{\rho \sim \mathcal{U}_{p}(\mathcal{D} \circ \AND_2)}
\Bigl[
\fourierwt{
\mathbb{E}_{x_{M(\rho)}}[g|_{\rho}]
}^{\ge k}
\Bigr]
\le
\sum_{t=k}^{pd}
\left(
p\sqrt{\frac{ed}{t}}
\right)^t.
\]

In particular, if
\(
k \ge 4ep^2d,
\)
then
\[
\mathbb{E}_{\rho \sim \mathcal{U}_{p}(\mathcal{D} \circ \AND_2)}
\Bigl[
\fourierwt{
\mathbb{E}_{x_{M(\rho)}}[g|_{\rho}]
}^{\ge k}
\Bigr]
\le
2^{-k+1}.
\]
\end{claim}

\begin{proof}
By construction, every restriction \( \rho \in \mathcal{D} \circ \AND_2 \)
fixes all \( x \)-variables outside the set \( V_x \) to \(1\).
Let \( h : \{0,1\}^{|V_x|} \to \{0,1\} \) be the Boolean function obtained from \( g \)
by fixing all variables outside \( V_x \) to \(1\).
Then for every \( \rho \in \mathcal{D} \circ \AND_2 \), we have \( g|_{\rho} = h \). Without loss of generality, assume $V_x = \{x_1,\ldots,x_d\}$.

Taking expectation over any variable kills all Fourier monomials containing it.
Therefore,
\[
\begin{aligned}
\mathbb{E}_{\rho \sim \mathcal{U}_{p}(\mathcal{D} \circ \AND_2)}
\Bigl[
\fourierwt{\,\mathbb{E}_{x_{M(\rho)}}\bigl[g|_{\rho}\,\bigr]}^{\ge k}
\Bigr] 
&=
\mathbb{E}_{\rho \sim \mathcal{U}_{p}(\mathcal{D} \circ \AND_2)}
\Bigl[
\fourierwt{\,\mathbb{E}_{x_{M(\rho)}}\bigl[h\,\bigr]}^{\ge k}
\Bigr] \\
&=
\sum_{\substack{S \subseteq [d]\\ |S|\ge k}}
|\widehat{h}(S)| \cdot
\Pr_{\rho}\bigl[ S \subseteq \rho^{-1}(*) \bigr].
\end{aligned}\]

Under the uniform distribution over restrictions with exactly \( pd \) variables marked by \( * \), the probability
\( \Pr_{\rho}[ S \subseteq \rho^{-1}(*) ] \)
is zero when \( |S| > pd \), and for \( |S| \le pd \) we have
\[
\Pr_{\rho}\bigl[ S \subseteq \rho^{-1}(*) \bigr]
=
\frac{\binom{d-|S|}{pd-|S|}}{\binom{d}{pd}}
\;\le\;
p^{|S|}.
\]
Therefore,
\[
\sum_{\substack{S \subseteq [d]\\ k \le |S| \le pd}}
|\widehat{h}(S)| \cdot
\Pr_{\rho}\bigl[ S \subseteq \rho^{-1}(*) \bigr]
\;\le\;
\sum_{t=k}^{pd}
\sum_{|S|=t}
|\widehat{h}(S)|\, p^t.
\]

Applying the Cauchy--Schwarz inequality, together with
Parseval’s identity (which implies
\( \sum_{S} \widehat{h}(S)^2 \le 1 \) for the Boolean function \( h \)),
we obtain
\(
\sum_{|S|=t} |\widehat{h}(S)| \le \sqrt{\binom{d}{t}}.
\)
Hence,
\[
\sum_{t=k}^{pd}
\sum_{|S|=t}
|\widehat{h}(S)|\, p^t \;\leq\;
\sum_{t=k}^{pd} \sqrt{\binom{d}{t}}\, p^t
\;\le\;
\sum_{t=k}^{pd}
\left(
p\sqrt{\frac{ed}{t}}
\right)^t.
\]

Finally, suppose \(k \ge 4ep^2d\).
Then for every \(t \ge k\),
\[
p\sqrt{\frac{ed}{t}}
\le
p\sqrt{\frac{ed}{k}}
\le
\frac12.
\]

Therefore,
\[
\sum_{t=k}^{pd}
\left(
p\sqrt{\frac{ed}{t}}
\right)^t
\le
\sum_{t=k}^{\infty} 2^{-t}
\le
2^{-k+1}.
\]

This completes the proof.
\end{proof}

\spargammatwo*

\begin{proof}
Let \( f \) have sparsity \( s \).
By \cref{lemma:spars-nonadapt-res-tree}, there exists a semi-adaptive
max-degree restriction tree \( \mathcal{D} \) for \( f \) of depth
\( d = c_1 \log s / \log n \), for a suitable absolute constant \( c_1 > 0 \).

Let \(c_2,c_3>0\) be absolute constants such that \cref{thm:error-reduction}
converts a \(0.44\)-approximator of degree \(r\) into a \(1/3\)-approximator
of degree at most \(c_2r\), and such that every Boolean function \(g\) satisfies
\(
\deg(g) \le c_3\cdot(\adeg(g))^2.
\)
Choose an absolute constant \(c>0\) sufficiently small so that
\(
c^{3/2} \le \frac{1}{4e c_2\sqrt{2c_3}}.
\)
Set
\[
p := c d^{-1/3}
\qquad\text{and}\qquad
k := \frac{c^{1/2}d^{1/3}}{c_2\sqrt{2c_3}}.
\]
Observe that \(pd = cd^{2/3}\), while
\(k=\Theta(d^{1/3})=\Theta(\sqrt{pd})\).

Now suppose, for the sake of contradiction, that there exists a decomposition
of the communication matrix \( M_{f \circ \AND_2} \) of the form
\[
\Pi(x,y) = \sum_{i=1}^m b_i\, g_i(x) h_i(y),
\]
where each \( g_i,h_i \) is Boolean, such that
\( \|\Pi - M_{f \circ \AND_2}\|_\infty \le 1/3 \)
and
\( \sum_{i=1}^m |b_i| \le \tfrac{1}{20}\,2^k \).
We derive a contradiction, thereby proving the theorem.

Let \( \mathcal{D} \circ \AND_2 \) be the collection of lifted restrictions for
\( f \circ \AND_2 \) obtained from \( \mathcal{D} \) via
\cref{algo:fourier-restriction}.
Sample a restriction
\( \rho \sim \mathcal{U}_p(\mathcal{D} \circ \AND_2) \).
We study the effect of applying \( \rho \) and then taking expectation over the
masked variables \( M(\rho) \). 
For notational convenience, define
\[
F_\rho := \mathbb{E}_{x_{M(\rho)}}\bigl[(f \circ \AND_2)|_{\rho}\bigr],
\qquad
G_{i,\rho} := \mathbb{E}_{x_{M(\rho)}}\bigl[g_i|_{\rho}\bigr].
\]
Both are functions of the starred \( x \)-variables under \( \rho \).

\begin{enumerate}
\item \textbf{Hardness of the restricted function.}
By \cref{claim:liftedrho-properties},
\(
\deg(F_\rho) = |\rho^{-1}(*)|.
\)
Since every restriction in the support of
\( \mathcal{U}_p(\mathcal{D} \circ \AND_2) \) satisfies
\( |\rho^{-1}(*)| = pd \),
we have
\(
\deg(F_\rho)=cd^{2/3}.
\)

\item \textbf{Simplification of the approximator \( \Pi \).}
For every restriction \( \rho \in \mathcal{D} \circ \AND_2 \), all \( y \)-variables
are fixed. As a result, for each term in the decomposition
\(
\Pi(x,y) = \sum_{i=1}^m b_i\, g_i(x) h_i(y),
\)
the restricted function \( h_i|_{\rho} \) becomes a constant, which we denote by
\( a_{\rho,i} \in \{0,1\} \). 
Thus,
\(
\Pi|_{\rho} = \sum_{i=1}^m b_i\, a_{\rho,i}\, g_i|_{\rho},
\)
which is a function only of the core \( x \)-variables, since all
non-core \( x \)-variables are fixed to \( 1 \) in every $\rho \in \mathcal{D} \circ \AND_2$.

Fix a restriction \( \rho \). 
The \( \ell_1 \)-mass of the Fourier spectrum of \( \mathbb{E}_{x_{M(\rho)}}\bigl[\Pi|_{\rho}\bigr] \) above level \( k \) can be bounded as follows:
\begin{align*}
\fourierwt{\mathbb{E}_{x_{M(\rho)}}[\Pi|_{\rho}]}^{\ge k}
&=
\sum_{\substack{S \subseteq [n]\\ |S| \ge k}}\Bigl|
\sum_{i=1}^m b_i\, a_{\rho,i}\, \widehat{G_{i,\rho}}(S)
\Bigr| \\
&\leq 
\sum_{\substack{S \subseteq [n]\\ |S| \ge k}} \sum_{i=1}^m |b_i||\widehat{G_{i,\rho}}(S)| \\
&=
\sum_{i=1}^m |b_i|
\sum_{\substack{S \subseteq [n]\\ |S| \ge k}} |\widehat{G_{i,\rho}}(S)| \\
&=
\sum_{i=1}^m |b_i|\, \fourierwt{G_{i,\rho}}^{\ge k}.
\end{align*}

Moreover,
\[
4ep^2d = 4ec^2d^{1/3}
\le
\frac{c^{1/2}d^{1/3}}{c_2\sqrt{2c_3}}
=
k,
\]
where the inequality follows from the choice of \(c\).
Taking expectation over \(\rho\) and applying \cref{claim:prop2}, we obtain

\[
\begin{aligned}
\mathbb{E}_{\rho \sim \mathcal{U}_p(\mathcal{D} \circ \AND_2)}
\Bigl[\fourierwt{\mathbb{E}_{x_{M(\rho)}}\bigl[\Pi|_{\rho}\bigr]}^{\ge k}
\Bigr] &\leq
\sum_{i=1}^m |b_i|\,
\mathbb{E}_{\rho \sim \mathcal{U}_p(\mathcal{D} \circ \AND_2)}
\Bigl[\fourierwt{G_{i,\rho}}^{\ge k}
\Bigr]\\
&\leq \sum_{i=1}^m |b_i| \cdot 2^{-k+1}
\;\leq\;
\frac{1}{10},
\end{aligned}\]

where the final inequality follows from the assumed bound on
\( \sum_i |b_i| \).

Thus, there exists a restriction \( \rho \sim \mathcal{U}_p(\mathcal{D} \circ \AND_2) \)
such that the \( \ell_1 \)-mass of the Fourier spectrum of
\( \Pi|_{\rho} \) above level \( k \), after averaging over the masked variables
\( M(\rho) \), is at most \(0.1\).

\end{enumerate}

Combining the two items, there exists a restriction \( \rho \) such that
\(\deg(F_\rho)=pd=cd^{2/3}\),
while the \( \ell_1 \)-mass of the Fourier spectrum of
\( \mathbb{E}_{x_{M(\rho)}}[\Pi|_{\rho}] \) above level \( k \) is at most \(0.1\).

Let \( \tilde{\Pi} \) be the polynomial obtained from
\( \mathbb{E}_{x_{M(\rho)}}[\Pi|_{\rho}] \) by deleting all Fourier monomials of
degree at least \( k \).
Since the discarded Fourier mass is at most \(0.1\) and
\( \Pi \) is a \(1/3\)-approximator of \( f \circ \AND_2 \),
the polynomial \( \tilde{\Pi} \) has degree \(<k\) and
\(0.44\)-approximates
 \( F_\rho \).
By standard error reduction (\cref{thm:error-reduction}), \( \tilde{\Pi} \) can be converted into a
\(1/3\)-approximator of degree at most \( c_2 k \).
As a result,
\[
\adeg(F_\rho)
<
c_2k
=
\frac{c^{1/2}d^{1/3}}{\sqrt{2c_3}}.
\]
On the other hand, \(
\deg(F_\rho) = cd^{2/3}\)
which contradicts the general inequality
\( \deg(g) \le c_3 \cdot (\adeg(g))^2 \)
for Boolean functions \( g \) (\cref{thm:degree-apxdegree}).
This contradiction completes the proof.
\end{proof}

\subsection{Consequences}
\label{subsec:consequences}

Knop et al.~\cite{knop2021log} showed that \( \log \spar(f) \) characterizes the
deterministic communication complexity of \(\AND\)-functions (\(f \circ \AND_2\)), up to polynomial loss and polylogarithmic factors in \( n \).
In particular, they proved that for every Boolean function \( f \),
\[
\Dcc(f \circ \AND_2)
= O\!\left((\log \spar(f))^5 \cdot \log n\right).
\]

Combining this bound with \cref{thm:spar-gamma_2}, together with the fact
that the logarithm of the approximate \( \gamma_2 \) norm lower bounds bounded-error
quantum communication complexity, we immediately obtain that for every Boolean
function \( f \),
\[
\Dcc(f \circ \AND_2)
=
O\!\left(\Qcc(f \circ \AND_2)^{15} \cdot (\log n)^6\right).
\]

A tighter relationship can be obtained using a more refined structural result of
Knop et al., which relates deterministic \(\AND\)-query complexity to sparsity and
a combinatorial measure known as \emph{monotone block sensitivity}.

\begin{definition}[Monotone Block Sensitivity]
The \emph{monotone block sensitivity} of a Boolean function \( f : \{0,1\}^n \to \{0,1\} \), denoted \( \mbs(f) \), is a variant of block sensitivity that only considers flipping 0’s to 1’s. A subset \( B \subseteq [n] \) is called a \emph{sensitive 0-block of \( f \) at input \( x \)} if \( x_i = 0 \) for all \( i \in B \), and \( f(x) \neq f(x \oplus 1_B) \), where \( x \oplus 1_B \) denotes the input obtained by flipping all bits in \( B \) from 0 to 1. For an input \( x \in \{0,1\}^n \), let \( \mbs(f,x) \) denote the maximum number of pairwise disjoint sensitive 0-blocks of \( f \) at \( x \). Then, \(\mbs(f) = \max_{x \in \{0,1\}^n} \mbs(f,x).
\)
\end{definition}

Knop et al.~\cite{knop2021log} showed that deterministic communication complexity
can be bounded in terms of both sparsity and monotone block sensitivity.

\begin{claim}[{\cite[Lemma~3.2, Claim~4.4, Lemma~4.6, Theorem~1.2]{knop2021log}}]
\label{claim:dadt-mbs-spar}
For every Boolean function \( f : \{0,1\}^n \to \{0,1\} \),
\[
\Dcc(f\circ \AND_2)
= O\!\left(\mbs(f)^2 \cdot \log \spar(f) \cdot \log n\right).
\]
\end{claim}

Intuitively, a large value of \( \mbs(f) \) indicates that a large-arity
\( \pOR \) function can be embedded into \( f \) via suitable restrictions and
identifications of variables.
When such a function \( f \) is lifted via composition with \( \AND_2 \), this
structure gives rise to an embedded instance of the \emph{unique set disjointness}
problem.

The \emph{unique set disjointness} function \( \UDISJ_k \) is a partial Boolean
function on inputs \( x,y \in \{0,1\}^k \), defined as
\[
\UDISJ_k(x,y) =
\begin{cases}
0, & \text{if } |x \land y| = 0, \\
1, & \text{if } |x \land y| = 1, \\
\text{undefined}, & \text{otherwise},
\end{cases}
\]
where \( x \land y \) denotes the bitwise AND and \( |\cdot| \) the Hamming weight. That is, under the promise that the inputs
are either bitwise disjoint or intersect in exactly one coordinate, the function distinguishes between these two cases.

The following result of Knop et al.~shows that large monotone block sensitivity
forces large embedded instances of unique set disjointness.

\begin{claim}[{\cite[Claim~4.7]{knop2021log}}]
\label{claim:mbs-udisj}
Let \( f : \{0,1\}^n \to \{0,1\} \) be a Boolean function with
\( \mbs(f) = k \).
Then the communication matrix of \( f \circ \AND_2 \) contains, as a submatrix
(up to flipping output bits), the communication matrix of
\( \UDISJ_k \).
\end{claim}
Using known lower bounds for unique set disjointness, we obtain the following.

\begin{theorem}[\cite{kalyanasundaram1992probabilistic,razborov1990distributional,razborov2003quantum,sherstov2009pattern}]
\label{thm:udisj}
\(
\Rcc(\UDISJ_k) = \Omega(k)
\text{ and }
\Qcc(\UDISJ_k) = \Omega(\sqrt{k}).
\)
\end{theorem}

\begin{claim}
\label{claim:mbs-qcc-rcc}
If \( f : \{0,1\}^n \to \{0,1\} \) satisfies \( \mbs(f) = k \), then
\[
\Rcc(f \circ \AND_2) = \Omega(k)
\qquad\text{and}\qquad
\Qcc(f \circ \AND_2) = \Omega(\sqrt{k}).
\]
\end{claim}

\begin{proof}
This follows immediately from
\cref{claim:mbs-udisj} and \cref{thm:udisj}.
\end{proof}

Combining \cref{claim:mbs-qcc-rcc} with
\cref{claim:dadt-mbs-spar}, we obtain the following relationships between
deterministic, randomized, and quantum communication complexity for
\(\AND\)-functions.

\lecandfns*
\begin{proof}
Combine \cref{thm:spar-gamma_2}, \cref{thm:agamma-qcc}, \cref{claim:dadt-mbs-spar} and \cref{claim:mbs-qcc-rcc}.
\end{proof}

\bibliographystyle{alpha}
\bibliography{BoolFnRefs}
\end{document}